\documentclass[12pt]{article}

\usepackage{amssymb}
\usepackage{amsmath}
\usepackage{amsthm}
\usepackage{fullpage}
\usepackage{epsfig}
\usepackage{float}
\def\lcm{{\rm lcm}}

\newtheorem{theorem}{Theorem}
\newtheorem{lemma}[theorem]{Lemma}
\newtheorem{corollary}[theorem]{Corollary}
\newtheorem{proposition}[theorem]{Proposition}

\newcommand{\be}{\begin{enumerate}}
\newcommand{\ee}{\end{enumerate}}
\newcommand{\ol}{\overline}

\newcommand{\txt}[1]{\mbox{ #1 }}

\author{Janusz Brzozowski, Jeffrey Shallit, and Zhi Xu\\
School of Computer Science\\
University of Waterloo\\
Waterloo, ON  N2L 3G1 \\
Canada\\
{\small \tt \{brzozo,shallit,z5xu\}@uwaterloo.ca}
}

\title{Decision Problems For Convex Languages}

\begin{document}

\maketitle

\begin{abstract}
In this paper we examine decision problems associated with various
classes of convex languages, studied by Ang and
Brzozowski (under the name ``continuous languages'').
We show that we can decide whether a given language $L$
is {prefix-,} suffix-, \hbox{factor-,} or subword-convex in polynomial
time if $L$ is represented by a DFA, but that the problem is 
PSPACE-hard if $L$ is represented by an NFA.  In the case that a regular
language
is not convex, we prove tight upper bounds on the length of the shortest
words demonstrating this fact, in terms of the number of states of an
accepting DFA.  Similar results are proved for some
subclasses of convex languages:  the
prefix-, suffix-, factor-, and subword-closed languages, and
the prefix-, suffix-, factor-, and subword-free languages.
\end{abstract}

\section{Introduction}

Thierrin \cite{Thierrin:1973} introduced convex languages with
respect to the subword relation.
Ang and Brzozowski \cite{Ang&Brzozowski:2008} generalized this concept to
arbitrary relations.
For example, a language $L$ is
said to be {\it prefix-convex\/} if, whenever $u, w \in L$ with
$u$ a prefix of $w$, then any word $v$ must also be in $L$ if
$u$ is a prefix of $v$ and $v$ is a prefix of $w$.
Similar definitions hold for suffix-, factor-, and
subword-convex languages.  (In this paper,
a ``factor'' is a contiguous block inside another
word, while a ``subword'' need not be contiguous. 
In the literature, these concepts are
sometimes called ``subword'' and ``subsequence'', respectively.)

A language is said to be {\it prefix-free\/} if whenever $w \in L$, then
no proper prefix of $w$ is in~$L$.  (By {\it proper\/} we mean a prefix
of $w$ other than $w$ itself.)
Prefix-free languages (prefix codes) were studied by Berstel and 
Perrin \cite{Berstel&Perrin:1985}.
Han has recently considered  $X$-free languages for various  values of $X$, such
as prefix, suffix, factor and subword~\cite{Han:2007}.

A language is said to be {\it prefix-closed\/} if whenever $w\in L$, then
every prefix of $w$ is also in~$L$.  Analogous definitions hold for
suffix-, factor-, and subword-closed languages.  A~factor-closed language
is often called {\it factorial\/}.

In this paper we consider the computational complexity of testing whether
a given language has the property of being prefix-convex, suffix-convex,
etc., prefix-closed, suffix-closed, etc., for a total of 12 different
problems.  
As we will see, the computational complexity of these decision problems
depends on how the language is represented.  If it is represented as the
language accepted by a DFA, then the decision problem is solvable in 
polynomial time.  On the other hand, if it is represented as a regular
expression or an NFA, then the decision problem is PSPACE-complete.  We also
consider the following question:  given that a language is {\it not\/}
prefix-convex, suffix-convex, etc., what is a good upper bound on the
shortest words ({\it shortest witnesses\/}) demonstrating this fact?

The remainder of the paper is structured as follows. In
Section~\ref{convexity} we study the complexity of testing for
convexity for languages represented by DFA's, and we include testing
for closure and freeness as special cases.  In Section~\ref{witnesses}
we exhibit shortest witnesses to the failure of the convexity
property.  Convex languages specified by NFA's are studied in
Section~\ref{NFAs}.  We also briefly consider convex languages
specified by context-free grammars in Section~\ref{CFGs}.
Section~\ref{conclusions} concludes the paper.

\section{Deciding convexity for DFA's}
\label{convexity}

     We will show that,
if a regular language $L$
is represented by a DFA $M$ with $n$ states,
it is possible to test the property of prefix-, suffix-, factor-, and
subword-convexity efficiently.  More precisely, we can test these properties
in $O(n^3)$ time.

Let $\unlhd$ be one of the four relations prefix, suffix,
factor, or subword.
The basic idea is as follows:  $L$ is {\it
not\/} $\unlhd$-convex if and only if there exist words
$u, w \in L$, $v \not\in L$, such that $u \unlhd v \unlhd w$.  Given $M$, we
create an NFA-$\epsilon$ $M'$ with $O(n^3)$ states and
transitions that accepts the language
$$ \lbrace w \in L(M) \ : \ \text{there exist } u \in L(M),
v \not\in L(M) \text{ such that } u \unlhd v \unlhd w \rbrace.$$
Then $L(M') = \emptyset$ if and only if $L(M)$ is $\unlhd$-convex.
We can test the emptiness of $L(M')$ using depth-first search in time
linear in the size of $M'$.  This gives an $O(n^3)$ algorithm for
testing the $\unlhd$-convex property.

     Since the constructions for all four properties are similar,
in the next subsection we handle the hardest case (factor-convexity) in
detail.  In the following subsections we content ourselves with a brief
sketch of the necessary constructions.

\subsection{Factor-convexity}
\label{factorc}

      Suppose $M = (Q, \Sigma, \delta, q_0, F)$ is a DFA accepting
the language $L = L(M)$, and suppose $M$ has $n$ states.  
We now construct an NFA-$\epsilon$ $M'$ such that
\begin{eqnarray*}
L(M') &=& \lbrace w \in \Sigma^* \ : \ \text{there exist} \ 
u, v \in \Sigma^* \ \text{such that} \  u \ \text{is a factor of}\  v, \\
&&
v \ \text{is a factor of}\ w, \text{and}\ u, w \in L, v \not\in L \rbrace.
\end{eqnarray*}
Clearly $L(M') = \emptyset$ if and only if $L(M)$ is factor-convex.

     Here is the construction of $M'$.  States of $M'$ are
quadruples, where components $1$, $2$, and $3$ keep 
track of where $M$ is upon
processing $w$, $v$, and $u$ (respectively).  The last component is a flag
indicating the present {\em mode\/} of the simulation process.
\medskip

     Formally, $M' = (Q', \Sigma, \delta', q'_0, F')$, where
\begin{eqnarray*}
Q' &=& Q \times Q \times Q \times \lbrace 1,2,3,4,5 \rbrace; \\
q'_0 &=& [q_0, q_0, q_0, 1]; \\
F' &=& F \times (Q-F) \times F \times \lbrace 5 \rbrace;\\
\mbox{1. }  \delta'([p, q_0, q_0, 1], a) &=& \lbrace [\delta(p,a), q_0, q_0, 1 ] \rbrace,
	\mbox{ for all } p \in Q, a \in \Sigma;\\
\mbox{2. }\hspace{.051cm} \delta'([p, q_0, q_0, 1], \epsilon) &=& \lbrace [p, q_0, q_0, 2] \rbrace,
	\mbox{ for all } p \in Q;\\
\mbox{3. }\hspace{.14cm} \delta'([p, q, q_0, 2], a) &=& \lbrace [ \delta(p,a), \delta(q,a), q_0, 2] \rbrace, 
	   \mbox{ for all } p, q \in Q, a \in \Sigma;\\
\mbox{4. }\hspace{.18cm} \delta'([p, q, q_0, 2], \epsilon) &=& \lbrace [ p, q, q_0, 3] \rbrace,
	\mbox{ for all } p, q \in Q;\\
\mbox{5. }\hspace{.28cm}\delta'([p, q, r, 3], a) &=& \lbrace [\delta(p,a), \delta(q,a),\delta(r,a), 3] \rbrace, 
	 \mbox{ for all } p, q, r \in Q, a \in \Sigma;\\
\mbox{6. }\hspace{.32cm} \delta'([p,q,r,3],\epsilon) &=& \lbrace [p,q,r,4] \rbrace,
	 \mbox{ for all } p, q, r \in Q;\\
\mbox{7. }\hspace{.28cm} \delta'([p,q,r,4], a) &=& \lbrace[\delta(p,a),\delta(q,a), r, 4] \rbrace,
	\mbox{ for all } p, q, r \in Q, a \in \Sigma;\\
\mbox{8. }\hspace{.32cm} \delta'([p,q,r,4], \epsilon) &=& \lbrace [p,q,r,5] \rbrace,
	\mbox{ for all } p, q, r \in Q;\\
\mbox{9. }\hspace{.28cm} \delta'([p,q,r,5], a) &=&  \lbrace [\delta(p,a), q, r, 5 ] \rbrace,
	\mbox{ for all } p, q, r \in Q, a \in \Sigma.
\end{eqnarray*}
One verifies that the NFA-$\epsilon$ $M'$ has $3n^3 + n^2 +n$ states and
$(3|\Sigma| + 2)n^3 + (|\Sigma|+1)(n^2 + n)$ transitions, where $|\Sigma|$ is the cardinality of $\Sigma$.

To see that the construction is correct, suppose $L$ is not factor-convex.
Then there exist words $u, v, w$ such that $u$ is a factor of $v$, $v$ is
a factor of $w$, and $u,w \in L$ while $v \not\in L$.
Then there exist words $u', u'', v', v''$ such that
such that $v=u'uu''$ and $w=v'vv''=v'u'uu''v''$.
Let $\delta(q_0,v')=q_1$,
$\delta(q_1, u')=q_2$,
$\delta(q_2,u)=q_3$,
$\delta(q_3,u'')=q_4$, and 
$\delta(q_4,v'')=q_5$. 
Moreover, let 
$\delta(q_0,u')=q_a$, $\delta(q_a,u)=q_b$, and $\delta(q_b,u'')=q_c$, and 
$\delta(q_0,u)=q_\alpha$.  Since $u, w \in L$, we know that $q_\alpha$ and
$q_5$ are accepting states.  Since $v \not\in L$, we know that
$q_c$ is not accepting.

Automaton $M'$ operates as follows.   In the initial state
$[q_0, q_0, q_0, 1]$ we process the symbols of $v'$ using Rule 1, ending in
the state $[q_1, q_0, q_0,1]$.  At this point, we use Rule 2 to move
to $[q_1, q_0, q_0,2]$ by an $\epsilon$-move.  Next, we process the
symbols of $u'$ using Rule 3, ending in the state
$[q_2, q_a, q_0, 2]$.  Then we use Rule 4 to move to
$[q_2, q_a, q_0, 3]$ by an $\epsilon$-move.  Next, we process the symbols
of $u$ using Rule 5, ending in the state $[q_3, q_b, q_\alpha, 3]$.
Then we use Rule 6 to move to $[q_3, q_b, q_\alpha, 4]$ by an
$\epsilon$-move.  Next, we process the symbols of $u''$ using Rule 7,
ending in the state $[q_4, q_c, q_\alpha, 4]$.  Then we use Rule 8 to move
to $[q_4, q_c, q_\alpha, 5]$ by an $\epsilon$-move.  Finally, we process the
symbols of $v''$ using 
Rule 9, ending in  the state $[q_5, q_c, q_\alpha, 5]$, and this state
is in $F'$.

On the other hand, suppose $M'$ accepts the input $w$.  Then 
we must have $\delta'(q'_0, w) \cap F'\not=\emptyset$.  But the only way to reach 
a state in $F'$ is, by our construction,
to apply Rules 1 through 9 in that order, where odd-numbered rules can
be used any number of times, and even-numbered rules can be used only
once.  Letting $v', u', u, u'', v''$ be the words labeling the uses of
Rules 1, 3, 5, 7, and 9, respectively, we see that
$w = v' u' u u'' v''$, where 
$\delta(q_0, w) \in L$, $\delta(q_0, u) \in L$, and
$\delta(q_0, u' u u'') \not\in L$.    It follows that
$u, w \in L$ and $v = u' u u'' \not\in L$, and so $L$ is not factor-convex.

\bigskip
We have proved

\begin{theorem}
\label{TDFAfactorcont}
     If $M$ is a DFA with $n$ states, there exists an NFA-$\epsilon$ $M'$
with $O(n^3)$ states and transitions such that $M'$ accepts the language
\begin{eqnarray*}
L(M') &=& \lbrace w \in \Sigma^* \ : \ \text{there exist} \ 
u, v \in \Sigma^* \ \text{such that} \  u \ \text{is a factor of}\  v, \\
&&
v \ \text{is a factor of}\ w, \text{and}\ u, w \in L, v \not\in L \rbrace.
\end{eqnarray*}
\end{theorem}

\begin{corollary}
\label{CDFAfactorcont}
We can decide if a given regular language $L$ accepted by a DFA with
$n$ states is factor-convex in $O(n^3)$ time.
\label{thm1}
\end{corollary}

\begin{proof}
Since $L$ is factor-convex if and only if $L(M') = \emptyset$,
it suffices to check if 
$L(M') = \emptyset$
using depth-first search of a directed graph, in time linear in the number
of vertices and edges of $M'$.  
\end{proof}

\subsubsection{Factor-closure}
\label{fac-clo-subsec}

The language
$L$ is {\it not\/} factor-closed if and only if there exist 
words $v, w$ such that $v$ is a factor of $w$,
and $w \in L$, while $v \not\in L$.    

Given a DFA $M$ accepting $L$,
we construct from $M$ an NFA-$\epsilon$ $M'$ such that
\begin{eqnarray*}
L(M') &=& \lbrace w \in \Sigma^* \ : \ \text{there exist} \ 
v, w \in \Sigma^* \ \text{such that} \ v \ \text{is a factor of}\ w,   \\
&&
 \text{and}\ w \in L, v \not\in L \rbrace.
\end{eqnarray*}
As before, $L(M') = \emptyset$ if and only if $L(M)$ is factor-closed.
The size of $M'$ is $O(n^2)$.  

States of $M'$ are
triples, where components $1$ and $2$ keep 
track of where $M$ would be upon
processing $w$, and $v$ (respectively).  The last component is a flag
as before.
\medskip

     Formally, $M' = (Q', \Sigma, \delta', q'_0, F')$, where
\begin{eqnarray*}
Q' &=& Q \times Q \times \lbrace 1,2,3 \rbrace; \\
q'_0 &=& [q_0, q_0, 1]; \\
F' &=& F \times (Q-F) \times \lbrace 3 \rbrace;
\quad\text{and}
\end{eqnarray*}
\be
\item
$\delta'([p, q_0, 1], a) = \lbrace [\delta(p,a), q_0, 1 ] \rbrace$
	for $p \in Q$, \ $a \in \Sigma$.

\item
$\delta'([p, q_0, 1], \epsilon) = \lbrace [p, q_0, 2] \rbrace$,
	for all $p \in Q$;
\item
$\delta'([p, q, 2], a) = \lbrace [ \delta(p,a), \delta(q,a),2]
	\rbrace$, for all $p, q \in Q$;
\item
$\delta'([p, q, 2], \epsilon) = \lbrace [ p, q, 3] \rbrace$,
	for all $p, q \in Q$;
\item
$\delta'([p,q,3], a) =  \lbrace [\delta(p,a), q, 3 ] \rbrace$,
	for $p, q \in Q$, $a \in \Sigma$.
\ee
$M'$ has $2n^2 +n$ states and
$(2|\Sigma| + 1)n^2 + (|\Sigma|+1)$ transitions.
Thus we have:

\begin{theorem}
\label{TDFAfactorclosure}
We can decide if a given regular language $L$ accepted by a DFA with
$n$ states is factor-closed in $O(n^2)$ time.
\label{thm11}
\end{theorem}

This result was previously obtained by B\'eal et 
al.\ \cite[Prop.\ 5.1, p.\ 13]{Beal&Crochemore&Mignosi&Restivo&Sciortino:2003}
through a slightly different approach.

The converse of the relation ``$u$ is a factor of $v$" is ``$v$
contains $u$ as a factor''.  This converse relation and similar
converse relations, derived from the prefix, suffix,  and subword
relations, lead to ``converse-closed
languages''~\cite{Ang&Brzozowski:2008}.  It has been shown by de Luca
and Varricchio~\cite{deLuca&Varricchio:1990} that a language $L$ is
factor-closed (factorial, in their terminology) if and only if it is a
complement of an ideal, that is, if and only if
$L = \overline{\Sigma^* K \Sigma^*}$ for some $K \subseteq \Sigma^*$.
Ang and Brzozowski~\cite{Ang&Brzozowski:2008} 
noted that a language is an ideal if and only
if it is converse-factor-closed, that is, if, for every $u\in L$, each
word of the form $v=xuy$ is also in $L$.  Thus, to test whether $L$ is
converse-factor-closed, we must check that there is no pair $(u,v)$
such that $u\in L$, $v\not\in L$, and $u$ is a factor of $v$.  This is
equivalent to testing whether $\overline{L}$ is factor-closed.  Then
the following  is an immediate consequence of
Theorem~\ref{TDFAfactorcont}:

\begin{corollary}
\label{CDFAconversefactorclosure}
We can decide if a given regular language $L$ accepted by a DFA with
$n$ states is an ideal in $O(n^2)$ time.
\end{corollary}

The results above also apply to other converse-closed languages.  Similarly,
any result about the size of witness demonstrating the lack of
prefix-, suffix- and subword-closure  apply also to the witness
demonstrating the lack of converse-prefix-, converse-suffix- and
converse-subword-closure, respectively.  Subword-closed and
converse-subword-closed languages were also investigated and
characterized by Thierrin \cite{Thierrin:1973}.

\subsubsection{Factor-freeness}
\label{Sssfactorfree}

Factor-free languages (also known as infix-free)
have recently been studied
by Han et al.\ \cite{Han&Wang&Wood:2006}; they gave an efficient algorithm
for determining if the language accepted by an NFA is
prefix-free, suffix-free, or factor-free.

We can decide whether a DFA language is factor-free in $O(n^2)$ time with the automaton we used for testing factor-closure, except that the set of accepting states is now 
$$F' = F \times F \times \lbrace 3 \rbrace.$$

Similar results hold for prefix-free, suffix-free, and subword-free languages.

\subsection{Prefix-convexity}
\label{prefixc}

     Prefix convexity can be tested in an analogous fashion.
We give the construction of $M'$ without proof:
let $M' = (Q', \Sigma, \delta', q'_0, F')$, where
\begin{eqnarray*}
Q' &=& Q \times Q \times Q \times \lbrace 1,2,3 \rbrace; \\
q'_0 &=& [q_0, q_0, q_0, 1]; \\
F' &=& F \times (Q-F) \times F \times \lbrace 3 \rbrace;\\
\delta'([p,q,r,1],a) &=& \lbrace [ \delta(p,a), \delta(q,a), \delta(r,a), 1 ]
\rbrace \ \quad \text{ for } p, q, r \in Q, \ a \in \Sigma; \\
\delta'([p,q,r,1], \epsilon) &=& \lbrace [ p,q,r,2] \rbrace  
\ \quad \text{ for } p, q, r \in Q; \\
\delta'([p,q,r,2], a) &=& \lbrace [ \delta(p,a), \delta(q,a), r, 2 ]
\rbrace \ \quad \text{ for } p, q, r \in Q, \ a \in \Sigma; \\
\delta'([p,q,r,2], \epsilon) &=& \lbrace [ p,q,r,3] \rbrace  
\ \quad \text{ for } p, q, r \in Q; \\
\delta'([p,q,r,3], a) &=& \lbrace [ \delta(p,a), q, r, 3 ]
\rbrace \ \quad \text{ for } p, q, r \in Q,\ a \in \Sigma.
\end{eqnarray*}
The NFA $M'$ has $3n^3$ states and $3(|\Sigma|+1)n^3$ transitions.

\subsubsection{Prefix-closure}

By varying the construction as before, we have

\begin{theorem}
\label{pswclosure}
We can decide if a given regular language $L$ accepted by a DFA with
$n$ states is prefix-closed, suffix-closed, or subword-closed in
$O(n^2)$ time.
\end{theorem}

\subsubsection{Prefix-freeness}
See Section~\ref{Sssfactorfree}.

\subsection{Suffix-convexity}
\label{suffixc}

     Suffix-convexity can be tested in an analogous fashion.
We give the construction of $M'$ without proof. 
Let $M' = (Q', \Sigma, \delta', q'_0, F')$, where
\begin{eqnarray*}
Q' &=& Q \times Q \times Q \times \lbrace 1,2,3 \rbrace ;\\
q'_0 &=& [q_0, q_0, q_0, 1] \rbrace; \\
F' &=& F \times (Q-F) \times F \times \lbrace 3 \rbrace;\\
\delta'([p,q_0,q_0,1],a) &=& \lbrace [ \delta(p,a), q_0, q_0, 1 ]
\rbrace \ \quad \text{ for } p \in Q, \ a \in \Sigma; \\
\delta'([p,q_0,q_0,1], \epsilon) &=& \lbrace [ p,q_0,q_0,2] \rbrace  
\ \quad \text{ for } p \in Q; \\
\delta'([p,q,q_0,2], a) &=& \lbrace [ \delta(p,a), \delta(q,a), q_0, 2 ]
\rbrace \ \quad \text{ for } p, q \in Q, \ a \in \Sigma; \\
\delta'([p,q,q_0,2], \epsilon) &=& \lbrace [ p,q,q_0,3] \rbrace  
\ \quad \text{ for } p, q \in Q; \\
\delta'([p,q,r,3], a) &=& \lbrace [ \delta(p,a), \delta(q,a), \delta(r,a), 3 ]
\rbrace \ \quad \text{ for } p, q, r \in Q,\ a \in \Sigma.
\end{eqnarray*}
The NFA $M'$ has $n^3$ states and $|\Sigma|n^3+(|\Sigma|+1)(n^2+n)$ transitions.

For results on suffix-closure and suffix-freeness, see
Theorem~\ref{pswclosure} and
Section~\ref{Sssfactorfree}, respectively.

\subsection{Subword-convexity}
\label{subwordc}

     Subword-convexity can be tested in an analogous fashion.
We give the construction of $M'$ without proof.
Let $M' = (Q', \Sigma, \delta', q'_0, F')$, where
\begin{eqnarray*}
Q' &=& Q \times Q \times Q; \\
q'_0 &=& [q_0, q_0, q_0]; \\
F' &=& F \times (Q-F) \times F;\\
\delta'([p,q,r],a) &= &\lbrace [\delta(p,a), q,r],\ 
	[\delta(p,a), \delta(q,a), r], \ 
	[\delta(p,a), \delta(q,a), \delta(r,a) ] \rbrace,\\
	&&\mbox{for all }  p, q, r \in Q \mbox{ and } a \in \Sigma.
\end{eqnarray*}
The NFA $M'$ has $n^3$ states and $|\Sigma|n^3$ transitions.

The idea is that as the symbols of $w$ are read, we keep track of the state
of $M$ in the first component.  We then ``guess'' which symbols of the 
input also belong to $u$ and/or $v$, enforcing the condition that, if
a symbol belongs to $u$, then it must belong to $v$, and if it belongs to
$v$, then it must belong to $w$.  We therefore cover all possibilities
of words $u, v$ such that $u$ is a subword of $v$ and $v$ is a subword of
$w$.

For results on subword-closure and subword-freeness, see
Theorem~\ref{pswclosure} and Section~\ref{Sssfactorfree},
respectively.

\subsection{Almost convex languages}
\label{almostcon}

     As we have seen, a language $L$ is prefix-convex if and only
if there are no triples $(u,v,w)$ with $u$ a prefix of $v$, $v$
a prefix of $w$, and $u, w \in L$, $v \not\in L$.  We call such a triple
a {\it witness}.  A language could fail
to be prefix-convex because there are infinitely many witnesses
(for example, the language $({\tt aa})^*$), or it could fail because
there is at least one, but only finitely many witnesses (for example,
the language $\epsilon + {\tt aa} {\tt a}^*$).   

     We define a language $L$ to be {\it almost prefix-convex\/} if
there exists at least one, but only finitely many witnesses to the failure
of the prefix-convex property.    Analogously, we define
{\it almost suffix-\/}, {\it almost factor-\/}, and
{\it almost subword-convex\/}.

\begin{theorem}
Let $L$ be a regular language accepted by a DFA with $n$ states. Then
we can determine if $L$ is almost prefix-convex (respectively, 
almost suffix-convex, almost factor-convex, almost
subword-convex) in $O(n^3)$ time.
\label{almostfc}
\end{theorem}

\begin{proof}
We give the proof for the almost factor-convex property, leaving the
other cases to the reader.

Consider the NFA-$\epsilon$ $M'$ defined in Section~\ref{factorc}.  As we have
seen, $M'$ accepts the language
\begin{eqnarray*}
L(M') &=& \lbrace w \in \Sigma^* \ : \ \text{there exist} \ 
u, v \in \Sigma^* \ \text{such that} \  u \ \text{is a factor of}\  v, \\
&&
v \ \text{is a factor of}\ w, \text{and}\ u, w \in L, v \not\in L \rbrace.
\end{eqnarray*}
Then $M'$ accepts an infinite language if and only if
$L$ is not almost factor-convex.   For if $M'$ accepts infinitely many
distinct words, then there are infinitely many distinct witnesses, while if
there are infinitely many distinct witnesses $(u,v,w)$, then there must be 
infinitely many distinct $w$ among them, since the lengths of $|u|$ and
$|v|$ are bounded by $|w|$.

Thus it suffices to see if $M'$ accepts an infinite language.  If $M'$ were
an NFA, this would be trivial:  first, we remove all states not reachable
from the start state or from which we cannot reach a final state.  
Next, we look for the existence of a cycle.  All three goals can be 
easily accomplished in time linear in the size of $M'$, using depth-first
search.  

However, $M'$ is an NFA-$\epsilon$, so there is one additional complication:
namely, that the cycle we find might be labeled completely by 
$\epsilon$-transitions.  To solve this, we use an idea suggested to us
by Jack Zhao and Timothy Chan (personal communication):  we find all the
connected components of the transition graph of $M'$ (which can be done in
linear time) and then, for each edge $(p,q)$ labeled with something other
than $\epsilon$ (corresponding to the transition $q \in \delta(p,a)$ for
some $a \in \Sigma$), we check to see if $p$ and $q$ are in the same
connected component.  If they are, we have found a cycle labeled with
something other than $\epsilon$.  This technique runs in linear time in the
size of the NFA-$\epsilon$.
\end{proof}

\subsubsection{Almost closed languages}

   In analogy with Section~\ref{almostcon}, we can define
a language $L$ to be {\it almost prefix-closed\/} if
there exists at least one, but only finitely many witnesses to the failure
of the prefix-closed property.    Analogously, we define
{\it almost suffix-\/}, {\it almost factor-\/}, and
{\it almost subword-closed\/}.

\begin{theorem}
Let $L$ be a regular language accepted by a DFA with $n$ states. Then
we can determine if $L$ is almost prefix-closed (respectively, 
almost suffix-closed, almost factor-closed, almost
subword-convex) in $O(n^2)$ time.
\label{almostfclosed}
\end{theorem}

\begin{proof}
Just like the proof of Theorem~\ref{almostfc}.
\end{proof}

\subsubsection{Almost free languages}

In a similar way, we can define
a language $L$ to be {\it almost prefix-free\/} if
there exists at least one, but only finitely many witnesses to the failure
of the prefix-free property.    Analogously, we define
{\it almost suffix-\/}, {\it almost factor-\/}, and
{\it almost subword-free\/}.

\begin{theorem}
Let $L$ be a regular language accepted by a DFA with $n$ states. Then
we can determine if $L$ is almost prefix-free (respectively, 
almost suffix-free, almost factor-free, almost
subword-free) in $O(n^2)$ time.
\label{almostffree}
\end{theorem}

\begin{proof}
Just like the proof of Theorem~\ref{almostfc}.
\end{proof}

\section{Minimal witnesses}
\label{witnesses}
Let $\unlhd$ represent one of the four relations: factor, prefix, suffix, or subword.
A necessary and sufficient condition that a language
$L$ be {\it not\/} $\unlhd$-convex is the existence of a triple 
$(u,v,w)$ of words, where $u, w \in L$, $v \not\in L$,
$u\unlhd v$, and $v\unlhd w$.
As before, we call such a triple a {\it witness\/}
to the lack of $\unlhd$-convexity.
A witness $(u,v,w)$ is {\it minimal\/} if every other witness $(u',v',w')$
satisfies $|w|<|w'|$, or $|w|=|w'|$ and $|v|< |v'|$,
or $|w|=|w'|$, $|v|= |v'|$, and $|u|< |u'|$.
The {\it size\/} of a witness is $|w|$. 

Similarly, if $L=L$ is not $\unlhd$-closed, then  $(v,w)$ is a {\it witness\/} if $w \in L$, $v \not\in L$, and $v\unlhd w$.
A witness $(v,w)$ is {\it minimal\/} if there exists no witness $(v',w')$  such that $|w'|<|w|$, or $|w'|=|w|$ and $|v'|< |v|$.
The {\it size\/}  is again $|w|$. 
For $\unlhd$-freeness witness, minimal witness, and size are defined as for $\unlhd$-closure, except that both words are in $L$.

Suppose we are given a regular language $L$ specified by an $n$-state DFA $M$,
and we know that $L$ is not $\unlhd$-convex (respectively,  $\unlhd$-closed or $\unlhd$-free). A natural question then is,
what is a good upper bound on the size of the
shortest witness that demonstrates the lack of this
property?

\subsection{Factor-convexity}

From Theorem~\ref{TDFAfactorcont},
we get
an $O(n^3)$ upper bound for a witness to the lack of factor-convexity.  

\begin{corollary}
Suppose $L$ is accepted by a DFA with $n$ states and $L$ is not
factor-convex.  Then there exists a witness $(u,v,w)$ such that $|w| \leq 3n^3 + n^2 + n-1$.
\label{cor1}
\end{corollary}

\begin{proof}
In our proof of Theorem~\ref{TDFAfactorcont}, we constructed an NFA-$\epsilon$
$M'$ with $3n^3 + n^2 + n$ states accepting
$
L(M') = \lbrace w \in \Sigma^* \ : \ \text{there exist} \ 
u, v \in \Sigma^* 
\ \text{such that} \  (u,v,w) \ \text{is a witness}\rbrace.
$
Thus, if $M$ is not factor-convex, $M'$ accepts such a word $w$, and the
length of $w$ is clearly bounded above by the number of states of
$M'$ minus $1$.
\end{proof}

    It turns out that the bound in Corollary~\ref{cor1} is best 
possible:

\begin{theorem}
There exists a class of non-factor-convex
regular languages $L_n$, accepted by DFA's with
$O(n)$ states, such the size of the minimal witness is 
$\Omega(n^3)$.
\label{non-factor-convex}
\end{theorem}

     The proof is postponed to Section~\ref{suff-convex-sec} below.

Results analogous to Corollary~\ref{cor1}
hold for prefix-, suffix-, and subword-convex languages.
However, in some cases we can do better, as we show below.  

\subsubsection{Factor-closure}
 \label{Sssecfactorclosure}

Theorem~\ref{TDFAfactorclosure} gives us a $O(n^2)$ upper bound on the
length of a witness to the failure of the factor-closed property:

\begin{corollary}
If $L$ is accepted by a DFA with $n$ states and $L$ is not
factor-closed, then there exists a witness $(v,w)$ such that $|w| \leq 2n^2+n-1$.
\label{cor11}
\end{corollary}

     It turns out that this $O(n^2)$ upper bound is best possible.
Let $M=(Q,\Sigma,\delta,q_0,F)$ be a DFA , where
$Q=\{q_0,q_1,\cdots,q_n,q_{n+1},p_0,p_1,\cdots,p_n,p_{n+1}\}$,
$\Sigma=\{ {\tt 0,1}\}$, $F=Q\setminus\{q_{n+1}\}$. For $1\leq i \leq n$,
$0\leq j\leq n$, the transition function is
\begin{eqnarray*}
\delta(q_0,{\tt 0}) &=& q_0, \\
\delta(q_0,{\tt 1}) &=& q_1, \\
\delta(q_i,{\tt 0}) &=& \begin{cases}
q_{i+1}, & \text{ if } i<n; \\
q_1, & \text{ if } i=n, 
\end{cases} \\
\delta(q_i,{\tt 1}) &=& \begin{cases}
q_1, & \text{ if }  i<n-1; \\
p_0, & \text{ if }  i=n-1; \\
q_{n+1}, & \text{ if }  i=n;
\end{cases} \\
\delta(q_{n+1}, {\tt 0}) &=& q_{n+1}, \\
\delta(q_{n+1}, {\tt 1}) &=& q_{n+1}, \\
\delta(p_j,{\tt 0}) &=& \begin{cases}
p_{j+1}, & \text{ if }  j<n; \\
q_0,  & \text{ if }  j=n;
\end{cases}\\
\delta(p_j,{\tt 1}) &=& \begin{cases}
q_{n+1}, & \text{ if }  j<n;\\
p_{n+1}, & \text{ if }  j=n;
\end{cases}\\
\delta(p_{n+1}, {\tt 0}) &=& q_{n+1}, \\
\delta(p_{n+1}, {\tt 1}) &=& q_{n+1}. \\
\end{eqnarray*}

    The DFA $M$ has $2n+4$ states. For $n=5$, $M$ is illustrated in Figure~\ref{cont3}. 

\begin{figure}[H]
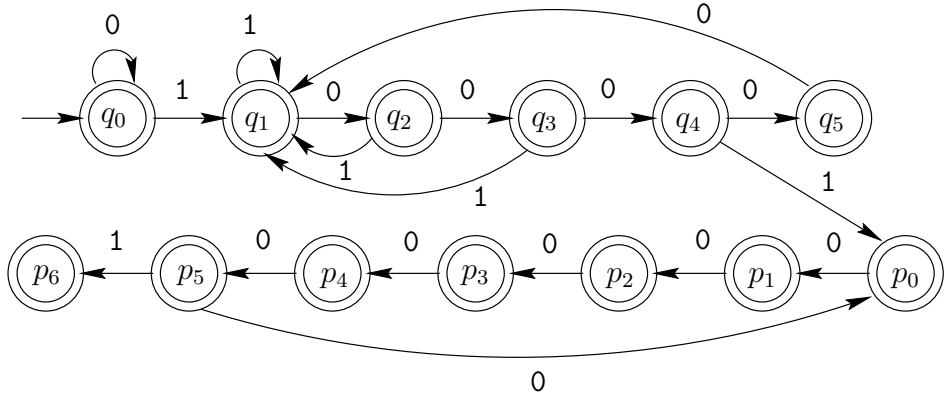

\begin{center}
\input cont3.pstex_t
\end{center}
\caption{Example of the construction in Theorem~\ref{xu} for
$n = 5$.  All unspecified transitions go to a rejecting
``dead state'' $q_6$ (not shown)
that cycles on all inputs.}
\label{cont3}
\end{figure}

Then we have the following theorem:

\begin{theorem}
For the  DFA $M$ above, let $L=L(M)$. For any
witness $(u,v)$ to the lack of factor-closure we have 
$|v|\geq (n+1)^2-1$, and this bound is achievable.
\label{xu}
\end{theorem}

\begin{proof}
Let $(u,v)$ be a minimal witness. Since  the
only rejecting state $q_{n+1}$ in $M$ leads only to itself, all the states
along the accepting path of $v$ are final. We claim that $u$ is a suffix of
$v$, that is,
$v=wu$ for some $w$. Otherwise, if the last letter of $v$ is not the last letter of $u$, we can just omit it and get a shorter $v$, which
contradicts the minimality of $v$. Similarly, all the states along the
rejecting path of $u$ except the last one are final; otherwise, we
get a shorter~$u$.

First, we prove that the set of states along the accepting path of
$v$ includes both $q$ states and $p$ states. Let $u={\tt 0}^i {\tt 1} u'$ for
$i\geq 0$. Then $\delta(q_1,u')=q_{n+1}$. If $\delta(q_0,w{\tt 0}^i)$ is a
$p$ state, we are done. Otherwise, let $\delta(q_0,
w {\tt 0}^i)=q_k$ for some $0\leq k\leq n$. If $k=n$, then
$\delta(q_0,v)=\delta(q_0, w {\tt 0}^i {\tt 1} u')
=\delta(q_k,{\tt 1} u')=\delta(q_n,{\tt 1} u')=\delta(q_{n+1},u')=q_{n+1}$,
a contradiction. If $k=n-1$, then
$\delta(q_0,w {\tt 0}^i {\tt 1})=\delta(q_k, {\tt 1})=p_0$,
which is a $p$ state.
Otherwise, $\delta(q_0,v)=\delta(q_k,{\tt 1} u')=\delta(q_1,u')=q_{n+1}$, a
contradiction. Hence, the set of states along the accepting path of
$v$ includes both $q$ states and $p$ states.

Now, consider the set of states along the rejecting path of $u$. We
prove that the set of states along the rejecting path of $u$ includes
only $q$ states. Suppose it includes both $q$ states and $p$ states.
Since there is only one transition from a $q$ state to a $p$ state
and all transitions from a $p$ state to a $q$ state are to the
rejecting state $q_{n+1}$, we have
$u=u_1u_2$, where $\delta(q_0,u_1)=q_{n-1}$, and
$$u_2\in L_1= {\tt 1}({\tt 0}^{n+1})^*(\epsilon+ {\tt 0} + {\tt 00} +
\cdots+{\tt 0}^{n-1}){\tt 1}.$$
Since $u$ is a suffix of $v$, the last letter of $v$ is also $\tt 1$. So,
by the construction of $M$, we have that $v=v_1v_2$, where
$\delta(q_0,v_1)=q_{n-1}$, and
$$v_2\in L_2= {\tt 1}({\tt 0}^{n+1})^* {\tt 0}^{n} {\tt 1}.$$
It is obvious that $(\Sigma^*L_1)\cap(\Sigma^*L_2)=\emptyset$, which
contradicts the equality $v_1v_2=v=wu=wu_1u_2$. Therefore, the set of states
along the rejecting path of $u$ includes only $q$ states.

Consider the last block of $\tt 0$'s in the words $u$ and $v$. By the structure
of $M$, we have
$$u\in\Sigma^* {\tt 1}({\tt 0}^n)^* {\tt 0}^{n-1} {\tt 1},$$
and
$$v\in\Sigma^* {\tt 1}({\tt 0}^{n+1})^*{\tt 0}^n {\tt 1}.$$
Therefore, the length of the last block of $\tt 0$'s is at
least $n(n+1)-1$. In other words, $|u|\geq n(n+1)-1+2=n^2+n+1$.
Since the shortest word that leads to  state $q_{n-1}$ (which is
the only state having a transition to a $p$ state on input $\tt 1$) is
${\tt 10}^{n-2}$, we also have $|v|\geq 1+n-2+n^2+n+1=n^2+2n$, and the first
part of this theorem proved.

To see that equality is achieved, let
$u={\tt 10}^{n^2+n-1} {\tt 1}$ and $v={\tt 10}^{n-2}u. $ 
\end{proof}

\subsubsection{Factor-freeness}

    From the remarks in Section~\ref{Sssfactorfree}, we get
\begin{corollary}
If $L$ is accepted by a DFA with $n$ states and $L$ is not
factor-free, then there exists a witness $(v,w)$ such that $|w| \leq 2n^2+n-1$.
\label{cor141}
\end{corollary}

    Up to a constant, Corollary~\ref{cor141} is best possible, as the following
theorem shows.

\begin{theorem}
There exists a class of languages accepted by DFA's with $O(n)$ states,
such that the smallest witness showing the language not factor-free
is of size $\Omega(n^2)$.
\end{theorem}

\begin{proof}
Let $L = {\tt bb}({\tt a}^n)^+{\tt b} \ \cup  \ {\tt b} ({\tt
a}^{n+1})^+{\tt b}$.  This language can be accepted by a DFA with
$2n+6$ states.  However, the shortest witness to lack of
factor-freeness is 
$({\tt b} {\tt a}^{n(n+1)}{\tt b}, {\tt bb} {\tt
a}^{n(n+1)}{\tt b})$,
which has size $n^2+n+3$.
\end{proof}

\subsection{Prefix-convexity}
For 
prefix-convexity, we have the following theorem.

\begin{theorem}
Let $M$ be a DFA with $n$ states.  Then if $L(M)$ is not prefix-convex,
there exists a witness $(u,v,w)$ with  $|w| \leq 2n-1$.  
Furthermore, this bound is best possible, as for all $n \geq 2$, there
exists a unary DFA with $n$ states that achieves this bound.
\label{npc}
\end{theorem}

\begin{proof}
If $L(M)$ is not prefix-convex, then such a witness $(u,v,w)$ exists.
Without loss of generality, assume that $(u,v,w)$ is minimal.         
Now write $w = uyz$, where $v = uy$ and $w = vz$.

Let $\delta(q_0, u) = p$, $\delta(p,y) = q$, and $\delta(q,z) = r$.
Let $P$ be the path from $q_0$ to $r$ traversed by $uvw$, and
let $P_1$ be the states from $q_0$ to $p$ (not including $p$),
$P_2$ be the states from $p$ to $q$ (not including $q$), and
$P_3$ be the states from $q$ to $r$ (not including $r$).
See Figure~\ref{cont1}.
Since $(u,v,w)$ is minimal, we know that every state of $P_3$ is rejecting, since we could have found a shorter $w$ if there were
an accepting state among them.
Similarly, every
state of $P_2$  must be accepting,
for, if there were a rejecting state among them,
we could have found a shorter $y$ and hence a shorter $v$.
Finally, every state of $P_1$ must be rejecting,
since, if there were an accepting state, we could have found a shorter $u$.

\begin{figure}[H]
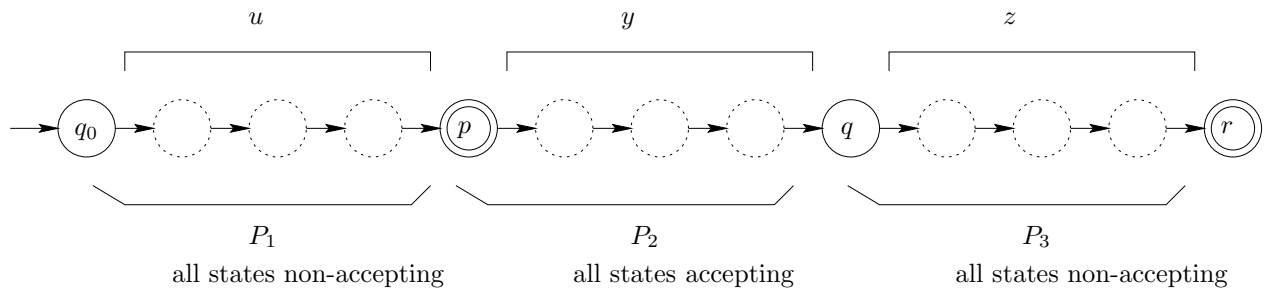

\begin{center}
\input cont1.pstex_t
\end{center}
\caption{The acceptance path for $w$}
\label{cont1}
\end{figure}

Let $r_i = |P_i|$ for $i = 1,2,3$.  
There are no repeated states in $P_3$, for if there were, we could
cut out the loop to get a shorter $w$; the same holds for $P_2$ and $P_1$.
Thus $r_i \leq n-1$ for $i = 1,2,3$.

Now $P_1$ and $P_2$ are disjoint, since all the states of $P_1$ are
rejecting, while all the states of $P_2$ are accepting.
Similarly, the states of $P_3$ are disjoint from $P_2$.
So $r_1 + r_2 \leq n$ and $r_2 + r_3 \leq n$.  
It follows that $r_1 + r_2 + r_3 \leq 2n-r_3$.  Since $r_3 \geq 1$, it
follows that $|w| \leq 2n-1$.  

To see that $2n-1$ is optimal, consider the DFA of $n$
states accepting the unary language $L = {\tt a}^{n-1}({\tt a}^n)^*$.  Then $L$
is not prefix-convex, and the shortest witness is
$({\tt a}^{n-1}, {\tt a}^n, {\tt a}^{2n-1})$.
\end{proof}

\subsubsection{Prefix-closure}

     For prefix-closed languages we can get an even better bound.

\begin{theorem}
Let $M$ be an $n$-state DFA, and suppose $L = L(M)$ is not prefix-closed.
Then the minimal witness $(v,w)$ showing $L$ is not prefix-closed
has $|w| \leq n$, and this is best possible.
\end{theorem}

\begin{proof}
Assume that $(v,w)$ is a minimal witness.
Consider the path $P$ from $q_0$ to $q = \delta(q_0, w)$, passing
through $p = \delta(q_0, v)$.  Let $P_1$ denote the part of the path $P$
from $q_0$ to $p$ (not including $p$) and $P_2$ denote the part of the
path from $p$ to $q$ (not including $q$).
Then all the states traversed in $P_2$
must be rejecting, because if any were
accepting we would get a shorter $w$.  Similarly, all the states traversed
in $P_1$
must be accepting, because otherwise
we could get a shorter $v$.  Neither $P_1$ nor $P_2$ contains a repeated
state, because if they did, we could ``cut out the loop'' to get a
shorter $v$ or $w$.  Furthermore, the states in $P_1$ are disjoint from
$P_2$.  So the total number of states in the path to $w$ (not counting $q$) 
is at most $n$.  Thus $|w| \leq n$.

The result is best possible, as the example of the unary language
$L = ({\tt a}^n)^*$ shows.   This language is not prefix-closed,
can be accepted by a DFA with
$n$ states, and the smallest witness is $({\tt a}, {\tt a}^n)$.
\end{proof}

\subsubsection{Prefix-freeness}
For the prefix-free property we have:

\begin{theorem}
If $L$ is accepted by a DFA with $n$ states and is not prefix-free,
then there exists a witness $(v,w)$ with $|w| \leq 2n-1$.  The bound
is best possible.
\end{theorem}

\begin{proof}
The proof is similar to that of Theorem~\ref{npc}.  The bound is achieved by
a unary DFA accepting
${\tt a}^{n-1} ({\tt a}^n)^*$.
\end{proof}

\subsection{Suffix-convexity}
\label{suff-convex-sec}
For the suffix-convex property, the cubic upper bound implied
by Corollary~\ref{cor1} is best possible, up to a constant factor.

\begin{theorem}
There exists a class of non-suffix-convex
regular languages $L_n$, accepted by DFA's with
$O(n)$ states, such the size of the minimal witness is 
$\Omega(n^3)$.
\label{nsc}
\end{theorem}

\begin{proof}
Let 
$$L =
{\tt bbb } ({\tt a}^{n-1})^+  \ \cup \ 
{\tt bb}({\tt a} + {\tt aa} + \cdots + {\tt a}^{n-1})  ({\tt a}^n)^* \ \cup \ 
{\tt b} ({\tt a}^{n+1})^+ .
$$
Then $L$ can be accepted by a DFA with $3n+5$ states, as illustrated in
Figure~\ref{cont2}.

\begin{figure}[H]
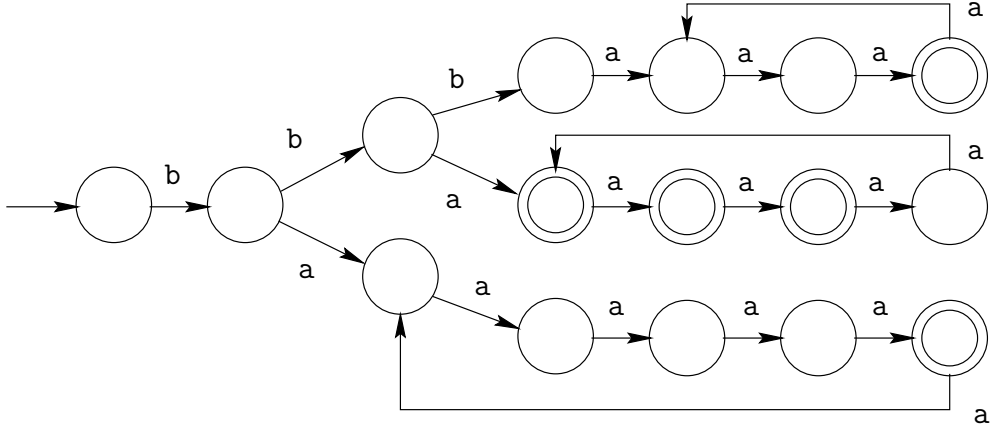

\begin{center}
\input cont2.pstex_t
\end{center}
\caption{Example of the construction in Theorem~\ref{nsc} for
$n = 4$.  All unspecified transitions go to a rejecting ``dead state'' $d$
that cycles on all inputs.}
\label{cont2}
\end{figure}

Suppose $(u, v, w)$ is a witness;  then $w$ cannot be a word of the form ${\tt b}{\tt a}^i$,
because no proper suffix of such a word is in $L$.
Also, $w$ cannot be a word of the form ${\tt bb} {\tt a}^i$, because the
only proper suffix in $L$ is $u = {\tt b} {\tt a}^i$.  But then there is
no word $v$ that lies strictly between $u$ and $w$ in the suffix order.  
So $w$ must be of the form ${\tt bbb} {\tt a}^i$.  The only proper suffixes
of $w$ in $L$ are of the form ${\tt bb} {\tt a}^i$ and
${\tt b}{\tt a}^i$.  But we cannot have $u = {\tt bb} {\tt a}^i$ because,
if we did, there would be no $v$ strictly between $u$ and $w$ in the suffix
order.  So it must be that $u = {\tt b} {\tt a}^i$.  Then the only word
in $\Sigma^*$ strictly between $u$ and $w$ in the suffix order is
$v = {\tt bb}{\tt a}^i$, and such a $v$ is not in $L$ if and only if
$i$ is a multiple of $n$.  On the other hand, for $u$ and $w$ to be in
$L$, $i$ must be a multiple of $n+1$ and $n-1$, respectively.

It follows that $L$ is not suffix-convex and the shortest witness
is $({\tt b} {\tt a}^i, {\tt bb} {\tt a}^i, {\tt bbb} {\tt a}^i)$,
where $i = \lcm (n-1,n,n+1) \geq (n-1)n(n+1)/2$.
\end{proof}

A similar technique can be used for non-factor-convex languages.
This allows us to prove Theorem~\ref{non-factor-convex}.

\begin{proof} (of Theorem~\ref{non-factor-convex})
Exactly like the proof of Theorem~\ref{nsc}, except we use the language
$L {\tt b}$ instead.
\end{proof}

\subsubsection{Suffix-closure}

Obviously, a witness to the failure of
suffix-closure is also a witness to the failure of factor-closure. 
So the proof of Theorem~\ref{xu} shows that  the bound $(n+1)^2-1$
also holds for suffix-closed languages.   

Ang and Brzozowski pointed out
\cite{Ang&Brzozowski:2008}
that a language $L$ is factor-closed
if and only if $L$ is both prefix-closed and suffix-closed. Here is
another relationship concerning the witnesses for these properties.

\begin{proposition}
Let $M$ be a DFA of $n$ states, and $L=L(M)$.
Let $v$ be the shortest word such that there is $u\not\in L, v\in L,
|v|>n$ and $u$ is a factor of $v$. Then $u$ is a suffix of $v$.
\end{proposition}

\begin{proof}Suppose $u$ is not a suffix of $v$. Write $v=v'a$
for $a\in\Sigma$. Then $u$ is also a factor of $v'$. So, $v'\not\in
L$. Since $|v'|\geq n$, by the pumping lemma, $v'=xyz$, such that
$xz\not\in L, |xz|<|xyz|$. But $xza\in L$ since $xyza=v'a=v\in L$.
This contradicts that $v$ is the shortest.
\end{proof}

In other words, a long minimal witness for factor-closure must
also be a witness for suffix-closure.

\subsubsection{Suffix-freeness}

\begin{theorem}
There exists a class of languages accepted by DFA's with $O(n)$ states,
such that the smallest witness showing the language not suffix-free
is of size $\Omega(n^2)$.
\label{notsuffixfree}
\end{theorem}

\begin{proof}
Let $L = {\tt bb}({\tt a}^n)^+ \ \cup  \ {\tt b} ({\tt a}^{n+1})^+$.
This language is accepted by a DFA with $2n+5$ states.
However, the shortest witness to the lack of suffix-freeness
$({\tt b} {\tt a}^{n(n+1)}, {\tt bb} {\tt a}^{n(n+1)})$ has size 
$n^2+n+2$.
\end{proof}

\subsection{Subword-convexity}
We now turn to subword properties. First, we recall some facts about the pumping lemma.
If $w = a_1  \cdots a_m$ with  $a_i\in \Sigma$ for $1 \leq i \leq m$,
we write $w[i,j]$ for the factor $a_i \cdots a_j$.
Assume that $M = (Q, \Sigma, \delta, q_0, F)$ is an $n$-state DFA, $m\ge n$, let $q\in Q$, and consider the state sequence $$S(q,w)=(\delta(q,w[1,0]),\ldots,\delta(q,w[1,m])).$$
We know that some state in $S(q,w)$ must appear more than once, because there are only $n$ distinct states in $M$. 
Let $\delta(q,w[1,i])$ be the first state that appears more than once in $S$, and let $x=w[1,i]$.
Moreover, let $\delta(q,w[1,j])$ be the first state in $S(q,w)$ equal to $\delta(q,w[1,i])$, and let $y=w[i+1,j]$. 
Finally, let $z=w[j+1,m]$.
Then $w=xyz$, where $|xy|\le n$ , $|y|>0$, and $|z|\ge m-n$, and $\delta(q,x)=\delta(q,xy)$.
By the pumping lemma,  $xy^*z\subseteq L$.
By the definition of $x$ and $y$, all the states in the sequence $S(q,w[1,j-1])$ are distinct. 
For a word $w$ with $|w|=m\ge n$, we refer to the factorization $w=xyz$  as the {\it canonical factorization  of $w$ with respect to $q$\/}.

\subsubsection{Subword-closure}
Here $v\unlhd w$ means $v$ is a subword of $w$.
If $L=L(M)$ is not subword-closed, then  $(v,w)$ is a {\it witness\/} if $w \in L$, $v \not\in L$, and $v\unlhd w$.

\begin{lemma}
\label{Lsubwordclosed}
Let $M$ be a DFA with $n\ge 2$ states such that $L(M)$ is not subword-closed.  For any  witness $(v,w)$, there exists a witness $(v',w')$ with $|w'| \leq n$ and $w'\unlhd w$.  
\end{lemma}
\begin{proof}
We will show that, for any witness $(v,w)$ with $|w|\ge n+1$, we can find a witness $(v',w')$ with 
$|w'|<|w|$ and $w'\unlhd w$. The lemma then follows.

Suppose that $(v,w)$ is a minimal witness, and $|w|=m\ge n+1$. 
Then the canonical factorization of $w$ is $w=xyz$,
where $|xy|\le n$, $|y|>0$, and $|z|\ge m-n>0$.

If there is a $z'$ such that $z'\unlhd z$ and $xyz'\not\in L$, then $xz' \not\in L$, since
$xyz'$ and $xz'$ lead to the same state in $M$. Then 
$(xz',xz)$ is
a witness with $|xz|< |w|$ and $xz\unlhd w$. 
Thus we can assume that
\begin{equation}
\label{eqnz'}
z'\unlhd z \txt{implies} xyz'\in L. 
\end{equation}
Since  $v\unlhd w=xyz$, we can write $v=v_xv_yv_z$, where 
$v_x\unlhd x$, $v_y\unlhd y$, and $v_z\unlhd z$.
Clearly, $v\unlhd xyv_z$.
If $v_z\not= z$, then by (\ref{eqnz'}), we have $xyv_z\in L$, and  $(v, xyv_z)$ is a witness with $|xyv_z|<|w|$ and $xyv_z\unlhd w$.
Thus we may assume that our witness has the form $(v_xv_yz,xyz)$.

In the particular case that $z'=\epsilon$, (\ref{eqnz'}) implies that $xy\in L$. 
If $y'\unlhd y$ and $xy'\not\in L$, then $(xy',xy)$ is a witness with $|xy|<|w|$ and $xy\unlhd w$.
Thus 
\begin{equation}
\label{eqny'}
y'\unlhd y \txt{implies} xy'\in L. 
\end{equation}

Finally, if $x'\unlhd x$ and $x'\not\in L$, then $(x',x)$ is a witness with $|x|<|w|$ and $x\unlhd w$.
Thus 
\begin{equation}
\label{eqnx'}
 x'\unlhd x \txt{implies} x'\in L. 
\end{equation}
Altogether, we
may assume that all the states along the path spelling $w$ in $M$ are accepting.
We know that the states in the sequence 
$$S=(\delta(q_0,w[1,0]), \ldots,\delta(q_0,w[1,|xy|-1]))$$
are all distinct. 
Also, the states in the sequence 
$$S'=(\delta(q_0,v_xv_yz[1,1]),\ldots,\delta(q_0,v_xv_yz[1,|z|-1]))$$ are all accepting and distinct; otherwise, $v$ would not be shortest.

We now claim that no state can be in both $S$ and $S'$. 
For suppose that 
$\delta(q_0,w[1,i])=\delta(q_0,v_xv_yz[1,k])$, for some $0\le i \le |x|$, $0< k < |z|$. 
Then $(w[1,i]z[k+1,|z|],xz)$ is a witness with $|xz|<|w|$ and $xz\unlhd w$, since $w[1,i]=x[1,i]$, and 
$x[1,i]z[k   +1,|z|]\unlhd xz$.
Next, if $\delta(q_0,xy[1,j])=\delta(q_0,v_xv_yz[1,k])$, for some $0< j <|y|$, $0< k < |z|$, 
then 
$$(xy[1,j]z[k+1,|z|], xyz[k+1,|z|])$$
 is a witness with $|xyz[k+1,|z|]|<|w|$ and
$xyz[k+1,|z|] \unlhd w$, since 
$xy[1,j]z[k+1,|z|]\unlhd xyz[k+1,|z|]$, and $xyz[k+1,|z|]\in L$ by~(\ref{eqnz'}).

Under these conditions $M$ must have $|xy|+(|z|-1)=|xyz|-1$ distinct accepting states, and at least one rejecting state. Hence $|xyz|=|w|\le n$ and we have found a witness with the required properties.
\end{proof}

\begin{corollary}
\label{Csubwordclosed}
Let $M$ be a DFA with $n\ge 2$ states.  If $L(M)$ is not subword-closed,
there exists a witness $(v,w)$ with $|w| \leq n$.  
Furthermore, this  is the best possible bound, as  there
exists a unary DFA with $n$ states that achieves this bound.
\end{corollary}

\begin{proof}
If $L$ is not subword-closed then it has a witness and, by Lemma~\ref{Lsubwordclosed}, it has a witness $(v,w)$ with $|w|\le n$.
This  is the best possible bound for $n\ge 2$, since the language
$$({\tt a}^n)^* (\epsilon +
{\tt a} + \cdots + {\tt a}^{n-2}),$$
accepted by a DFA
with $n$ states, has a minimal witness $({\tt a}^{n-1}, {\tt a}^n)$. 
\end{proof}

For $n=1$, $L$ is either $\emptyset$ or $\Sigma^*$, and both of these languages are subword-closed.

\subsubsection{Subword-freeness}
 
\begin{lemma}
\label{Lsubwordfree}
Let $M$ be a DFA with $n\ge 2$ states such that $L(M)$ is not
subword-free.  For any  witness $(u,w)$, there exists a witness
$(u',w')$ with $|w'| \leq 2n-1$, and $w'\unlhd w$.
\end{lemma}

\begin{proof}
We will show that, for any witness $(u,w)$ with $|w|\ge 2n$, we can
find a witness $(u',w')$ with $|w'|<|w|$ and $w'\unlhd w$. The lemma
then follows.

Let the canonical factorization of $w$ with respect to $q_0$ be
$w=xyz$, where $|xy|\le n$, $|y|>0$, and $|z|\ge n>0$.  Then we also
have a canonical factorization of $z=x'y'z'$ with respect to state
$q=\delta(q_0,xy)$, where $|x'y'|\le n$, $|y'|>0$, and $|z'|\ge 0$.
Now we have a witness $(xx'z',xx'y'z')=(xx'z',xz)$  with $|xz|<|w|$ and
$xz\unlhd w$.
\end{proof}

\begin{corollary}
\label{Csubwordfree}
Let $M$ be a DFA with $n\ge 2$ states.  If $L(M)$ is not subword-free,
there exists a witness $(u,w)$ with $|w| \leq 2n-1$.  
Furthermore, this  is the best possible bound, as  there
exists a unary DFA with $2n-1$ states that achieves this bound.
\end{corollary}
\begin{proof}
If $L$ is not subword-free then it has a witness and, by
Lemma~\ref{Lsubwordfree}, it has a witness $(v,w)$ with $|w|\le 2n-1$.
This  is the best possible bound for $n\ge 2$, since the language ${\tt
a}^{n-1} ({\tt a}^n)^*$, accepted by a DFA with $n$ states, has a
minimal witness $({\tt a}^{n-1},{\tt a}^{2n-1})$.  
\end{proof}

For
$n=1$, $L$ is either $\emptyset$ or $\Sigma^*$. Only $\Sigma^*$ is not
subword-free, and has a minimal witness $(\epsilon,{\tt a})$ for any
$a\in\Sigma$.

\subsubsection{Subword-convexity}

\begin{lemma}
\label{subwordconvex}
Let $M$ be a DFA with $n\ge 2$ states such that $L(M)$ is not
subword-convex.  For any  witness $(u,v,w)$, there exists a witness
$(u',v',w')$ with   $w'\unlhd w$, and $|w'| \leq 3n-2$.
\end{lemma}

\begin{proof}
We will show that, for any witness $(u,v,w)$ with $|w|\ge3n-1$, we can find a witness $(u',v',w')$ with 
$|w'|<|w|$ and $w'\unlhd w$. The lemma then follows.

We may assume without loss of generality
that $v$ is a shortest possible word corresponding to the
given $w$, and $u$ is a shortest word corresponding to $v$ and $w$.

First, consider the witness $(u,v)$ for lack of subword-closure of the
language $\ol{L}$.  By Lemma~\ref{Lsubwordclosed}, there exists a
witness $(u',v')$ to the failure of the subword-closure property
of $\ol{L}$ such that
$v'\unlhd v$ and $|v'|\le n$.  Therefore we can assume that we have a
witness $(u,v,w)$ to the failure of
subword-convexity such that $|v|\le n$.

Suppose that $(u,v,w)$ is a minimal witness, and $|w|\ge 3n-1$.
Then the canonical factorization of $w$ is $w=x_1y_1z_1$, where $|x_1y_1|\le n$, $|y_1|>0$, and 
$|z_1|\ge 2n-1\ge n> 0$.
Consider the states
$$p_0=\delta(q_0,x_1y_1), \ p_1=\delta(q_0,x_1y_1z_1[1,1]), \ \cdots,\ 
p_{|z_1|}=\delta(q_0,x_1y_1z_1).$$
Since $|z_1|\ge n$, there must be at least one pair $(p_i,p_j)$ of states such that $p_i=p_j$. 
If $p_0$ is the state that is repeated, let $i$ be the greatest index such that $p_0=p_i$, and let $x_2=\epsilon$, $y_2=z_1[1,i]$, and $z_2=z_1[i+1,|z_1|]$.
If $p_i$ is the first state that is repeated, let $j$ be the greatest index such that $p_i=p_j$, and let $x_2=z_1[1,i]$, $y_2=z_1[i+1,j]$, and $z_2=z_1[j+1,|z_1|]$.
If 
$\delta(q_0,x_1y_1x_2y_2), \delta(q_0,x_1y_1x_2y_2z_2[1,1]), \ldots,\delta(q_0,x_1y_1x_2y_2z_2)$
 has no repeated states, we stop. 
Otherwise, we apply the same procedure to $z_2$, and so on.
In any case, eventually we reach a $z_k$ for which no repeated states exist.
Then we have the factorization
$$w=x_1y_1x_2y_2\cdots x_ky_kz_k, $$
where $x_1y^*_1x_2y^*_2\cdots x_ky^*_kz_k
\subseteq L$,
$|x_2\cdots x_kz_k|<n$ (otherwise, there would be repeated states),  $|y_i|>0$, for $i=1,\ldots,k$, and $k\geq 2$.

For any $y_2'\unlhd y_2, \cdots, y_k'\unlhd y_k$, we have
$x_1y_1x_2y_2'\cdots x_ky_k'z_k \in L.$
Otherwise, the triple 
$$(x_1x_2\cdots x_kz_k, \;x_1x_2y_2'\cdots
x_ky_k'z_k, \; x_1x_2y_2\cdots x_ky_kz_k)$$
 is a witness with $|x_1x_2y_2\cdots x_ky_kz_k|< |w|$, and $x_1x_2y_2\cdots x_ky_kz_k\unlhd w$.

Since $v\unlhd w$, we can now write
$v=v_{x_1}v_{y_1}v_{x_2}v_{y_2}\cdots v_{x_k}v_{y_k}v_{z_k},$
where $v_{x_1}\unlhd x_1$, etc.
If there is a $y_i$ with $i\ge 2$, such that $v_{y_i}=\epsilon$, then we can replace that $y_i$ by $\epsilon$ in $w$ and obtain a smaller witness.
Hence each $v_{y_i}$ must be nonempty.
By the same argument, if there is a letter in $y_i$, for $i\ge 2$, that is not used in $v_{y_i}$, then that letter can be removed, yielding a smaller witness.
Therefore  $y_i=v_{y_i}$ for $i=2,\ldots,k$.
We claim that  $|y_2\cdots y_k|<|v|$; otherwise $v=v_{y_2}\cdots
v_{y_k}=y_2\cdots y_k$ and $(u,v, x_1x_2y_2\cdots x_ky_kz_k)$ is a
witness with $|x_1x_2y_2\cdots x_ky_kz_k|<|w|$.  Thus
$|y_2\cdots y_k|<|v|\le n$, and $$|w|=|x_1y_1|+|x_2\cdots
x_kz_k|+|y_2\cdots y_k|\leq n+(n-1)+(n-1)=3n-2.$$
\end{proof}

\begin{corollary}
\label{subwordcontcor}
Let $M$ be a DFA with $n\ge 2$ states.  If $L(M)$ is not subword-convex,
there exists a witness $(u,v,w)$ with 
 $|w| \leq 3n-2$.  
\end{corollary}

We do not know whether $3n-2$ is the best bound.
The unary language
${\tt a}^{n-1} ({\tt a}^n)^*$ is accepted by a DFA with $n$
states and has a minimal witness $({\tt a}^{n-1},{\tt a}^n,{\tt a}^{2n-1}$), showing that $2n-1$ is achievable.

\section{Languages specified by NFA's}
\label{NFAs}
     In this section consider some of the same problems as we have for
DFA's in previous sections.  
\subsection{Deciding convexity for NFA's}
Our  main result is that some of our decision
problems
become PSPACE-complete
if $M$ is represented by an NFA.  Our fundamental tool is the following
classical lemma \cite{Aho&Hopcroft&Ullman:1974}:

\begin{lemma}
Let $T$ be a one-tape deterministic Turing machine and $p(n)$ a polynomial
such that $T$ never uses more than $p(|x|)$ space on input $x$.  Then
there is a finite alphabet $\Delta$ and a polynomial $q(n)$ such that we
can construct a regular expression $r_x$ in $q(|x|)$ steps, such that
$L(r_x) = \Delta^*$ if $T$ doesn't accept $x$, and 
$L(r_x) = \Delta^* - \lbrace w \rbrace$ for some nonempty $w$ (depending on $x$)
otherwise.  Similarly, we can construct an NFA $M_x$ in $q(|x|)$ steps,
such that $L(M_x) = \Delta^*$ if $T$ doesn't accept $x$, and $L(M_x) =
\Delta^* - \lbrace w \rbrace$ for some nonempty $w$ (depending on $x$)
otherwise.
\end{lemma}

\begin{theorem}
     The problem of deciding whether a given regular language $L$,
represented by an NFA or regular expression, is prefix-convex
(resp., suffix-, factor-, subword-convex), or
prefix-closed (resp., suffix-, factor-, subword-closed) is PSPACE-complete.
\end{theorem}

\begin{proof}
     We prove the result for factor-convexity, the other results being
proved in the same way.

     First, let's see that the problem of deciding factor-convexity
is in PSPACE.  We actually
show that we can solve it in NSPACE, and then use Savitch's theorem
that PSPACE = NSPACE.

     Suppose $L$ is accepted by an NFA $M$ with $n$ states.  Then, by
the subset construction, $L$ is accepted by a DFA with $\leq 2^n$ states.
From    
Theorem~\ref{thm1} above, we see that if $L$ is not factor-convex,
we can demonstrate this by exhibiting $u, v, w$ with $u$ a prefix of 
$v$ and $v$ is a prefix of $w$ , and $u, w \in L$, $v \not\in L$ and
then checking that these conditions are fulfilled.  Furthermore, 
from Corollary~\ref{cor1}, if such $u, v, w$ exist, then $|u|,|v|,|w|
= O((2^n)^3)$.    In polynomial space, we can count up to $2^{3n}$.  
Write $w = x_1 x_2 x_3 x_4 x_5$, and let $v = x_2 x_3 x_4$ and 
$u = x_3$.  We use boolean matrices to keep track of, for each state
of $M$, what state we would be in after reading prefixes of $w$.
We guess the appropriate words $x_1, x_2, x_3, x_4, x_5$ symbol-by-symbol,
using a counter to ensure these words are shorter than $2^{3n}$.
We then verify that $x_3 $ and $x_1 x_2 x_3 x_4 x_5$ are in $L$ and
$x_2 x_3 x_4$ is not.

     Now let's see that the problem is PSPACE-hard.  Since $\Delta^*$ is
factor-convex and
$\Delta^* - \lbrace w \rbrace$ is not if $w \not= \epsilon$, we could
use an algorithm solving the factor-convex problem 
to solve decidability for
polynomial-space bounded Turing machines.
\end{proof}

       However, the situation is different for deciding the property of
prefix-freeness, suffix-freeness, etc., for languages represented by
NFA's, as the following theorem shows.  This result was proved already
by Han et al.\ \cite{Han&Wang&Wood:2006} through a different approach.

\begin{theorem}
Let $M$ be an NFA with $n$ states and $t$ transitions.  Then we can 
decide in $O(n^2+t^2)$ time whether $L(M)$ is prefix-free (resp., suffix-free,
factor-free, subword-free).
\end{theorem}

\begin{proof}
We give the full details for prefix-freeness,  and sketch the result for
the other three cases.

Given $M = (Q, \Sigma, \delta, q_0, F)$,
create an NFA $M'$ accepting $L(M) \Sigma^+$.  This can be done,
for example, by adding a transition on each $a \in \Sigma$ from each old final
state of $M$ to a new state $q_f$, and having a loop on $q_f$ to itself
on each $a \in \Sigma$.    Finally, let the new set of final states for $M'$
be $\lbrace q_f \rbrace$.  Clearly that $L(M)$ is prefix-free if and only
if $L(M) \ \cap \ L(M') = \emptyset$.  We can construct an NFA $M''$ accepting
$L(M) \ \cap \ L(M')$ using the usual ``direct product'' construction.
If the original $M$ had $n$ states and $t$ transitions, the new $M'$
has $n+1$ states and at most $t + 2n |\Sigma|$ transitions.  So $M''$ has
$n(n+1)$ states and at most $t(t+2n|\Sigma|)$ transitions.  Since
without loss of generality we can assume that $t \geq n-1$ (otherwise $M$
is not connected), it costs $O(n^2+t^2)$ to check whether $L(M'') = \emptyset$
using depth-first search.

For suffix-freeness, we carry out a similar construction for
$L(M) \ \cap \ \Sigma^+ L(M)$.  For factor-freeness, we carry out a
similar construction for
$L(M) \ \cap \ (\Sigma^+ L(M) \Sigma^* \ \cup \ \Sigma^* L(M) \Sigma^+)$.

For subword-freeness, we carry out a similar but slightly more
involved construction, which is
as follows:  create $M'$ by making two copies of $M$.  Add a transition
from each state $q$ to its copy $q'$ on each letter of $\Sigma$, and
add transitions from each copy $q'$ to itself on each letter of $\Sigma$.
The final states of $M'$ are the final states in the part corresponding
to the copied states.    Formally,
$M' = (Q \cup Q', \Sigma, \delta, q_, F')$ where
$Q' = \lbrace q' \ : \ q \in Q \rbrace$,
$F' = \lbrace q' \ : \ q \in F \rbrace$, and
$\delta'(q,a) = \delta(q,a) \ \cup \ \lbrace q' \rbrace$
for all $q \in Q,\ a \in \Sigma$, and
$\delta'(q',a) = \delta(q,a)' \ \cup \ \lbrace q' \rbrace$
for all $q \in Q, a \in \Sigma$.  Then $M'$ accepts the language of
all words that are strict superwords of words accepted by $M$.
We now create the NFA for $L(M) \ \cap \ L(M')$ as before.
\end{proof}

\subsection{Minimal witnesses for NFA's}

     We have already seen that the length of the minimal witness for the failure
of the convex or closed properties is polynomial
in the size of the DFA.  For the case of NFA's, however, this
bound no longer holds.

\begin{theorem}
     There exists a class of NFA's with $O(n)$ states such that the
shortest witness to the failure of the prefix-convex (resp.,
suffix-convex, factor-convex, subword-convex) or
prefix-closed (resp., suffix-closed, factor-closed, subword-closed)
property is of length $2^{\Omega(n)}$.
\end{theorem}

\begin{proof}
In Ellul et al.\ \cite[\S 5, p.\ 433]{Ellul&Krawetz&Shallit&Wang:2005}
the authors show how to construct a regular expression $E$ of length $O(n)$
that accepts all words up to some length
$2^{\Omega(n)}$, at which point a string is omitted.  From $E$
one can construct an NFA with $O(n)$ states accepting an $L$ with the
desired property.
\end{proof}

    For the prefix-free, etc., properties, we have

\begin{theorem}
     There exists a class of languages, accepted by NFA's with
$O(n)$ states and $O(n)$ transitions, such that the minimal witness
for the failure of the prefix-free property is of length $\Omega(n^2)$.
\end{theorem}

\begin{proof}
    For non-prefix-free, we can use the reverse of the language
defined in the proof of Theorem~\ref{notsuffixfree}.
\end{proof}

   For the failure of the subword-free property, however, we cannot
improve the bound we obtained for DFA's in Corollary~\ref{Csubwordfree}, as the
proof we presented there also works for NFA's.

\section{Languages specified by context-free grammars}
\label{CFGs}

     If $L$ is represented by a context-free grammar, then the
decision problems corresponding to convex and closed languages
become undecidable.  This follows easily from a well-known result that
the set of invalid computations of a Turing machine is a CFL
\cite[Lemma 8.7, p.\ 203]{Hopcroft&Ullman:1979}.

      Similarly, the decision problems corresponding to the
properties of prefix-free, suffix-free, and factor-free become
undecidable for CFL's, as shown by J\"urgensen and Konstantinidis
\cite[Thm.\ 9.5, p.\ 581]{Jurgensen&Konstantinidis:1997}.

      However, testing subword-freeness is still decidable for
CFL's:

\begin{theorem}
      There is an algorithm that, given a context-free grammar $G$, will
decide if $L(G)$ is subword-free.
\end{theorem}

\begin{proof}
     If $L = L(G)$ is infinite, then $L$ is not subword-free by the
pumping lemma.  For if $|w|$ is sufficiently large, then we can
factor $w$ as $uvxyz$, where $|vy| \geq 1$, such that 
$uxz \in L$.  But $uxz$ is a subword of $w$.
We can test if $L(G)$ is infinite by a well-known result
\cite[Thm.\ 6.6, p.\ 137]{Hopcroft&Ullman:1979}.  Otherwise, if $L(G)$ is
finite, we can enumerate all its words and test each for the subword-free
property.
\end{proof}

\section{Conclusions}
\label{conclusions}

We have shown that we can decide in $O(n^3)$ time whether a language specified
by a DFA is prefix-, suffix-, factor-, or subword-convex, and that the 
corresponding closure and freeness properties can be tested in $O(n^2)$
time.  If the language is specified by an NFA or a regular expression,
these problems are PSPACE-complete.  

Our results about the sizes of minimal witnesses for the various classes are summarized in Table~\ref{witnesssummary}.
All results are known to be
best possible, except the $3n-2$ upper bound for subword-convexity; in this
case, we do not know whether the bound is achievable.

\begin{table}[H]
\caption{Sizes of witnesses}
    \label{witnesssummary}
  \begin{center}
$
\begin{array}{|l|c|c|c|} \hline
\text{property}                    &  \text{convexity}  & \text{closure}   & \text{freeness} \\
\text{relation}                    &   &  &  \\
\hline
\text{factor}  &  \Theta(n^3) &  \Theta(n^2)    &   \Theta(n^2)          \\
\text{prefix}  &  2n-1  &  n   &   2n-1      \\
\text{suffix}  & \Theta(n^3)  & \Theta(n^2)    &  \Theta(n^2)     \\
\text{subword}  &  3n-2  & n   &  2n-1   \\                                                                                           
\hline
\end{array}
$
  \end{center}
\end{table}


\begin{thebibliography}{10}

\bibitem{Aho&Hopcroft&Ullman:1974}
A.~Aho, J.~Hopcroft, and J.~Ullman.
\newblock {\em The Design and Analysis of Computer Algorithms}.
\newblock Addison-Wesley, 1974.

\bibitem{Ang&Brzozowski:2008}
T.~Ang and J.~Brzozowski.
\newblock Continuous languages.
\newblock In E.~Csuhaj-{Varj\'u} and Z.~{\'Esik}, editors, {\em Proc. 12th
  International Conference on Automata and Formal Languages}, pp.  74--85.
  Computer and Information Research Institute, Hungarian Academy of Sciences,
  2008.

\bibitem{Beal&Crochemore&Mignosi&Restivo&Sciortino:2003}
M.-P. {B\'eal}, M.~Crochemore, F.~Mignosi, A.~Restivo, and M.~Sciortino.
\newblock Computing forbidden words of regular languages.
\newblock {\em Fund. Inform.} {\bf 56} (2003), 121--135.

\bibitem{Berstel&Perrin:1985}
J.~Berstel and D.~Perrin.
\newblock {\em Theory of Codes}.
\newblock Academic Press, New York, 1985.

\bibitem{deLuca&Varricchio:1990}
A.~de Luca and S.~Varricchio.
\newblock Some combinatorial properties of factorial languages.
\newblock In R.~Capocelli, editor, {\em Sequences}, pp.  258--266. Springer,
  1990.

\bibitem{Ellul&Krawetz&Shallit&Wang:2005}
K.~Ellul, B.~Krawetz, J.~Shallit, and {Wang, M.-w.}
\newblock Regular expressions: new results and open problems.
\newblock {\em J. Automata, Languages, and Combinatorics} {\bf 10} (2005),
  407--437.

\bibitem{Han:2007}
Y.-S. Han.
\newblock Decision algorithms for subfamilies of regular languages using
  state-pair graphs.
\newblock {\em Bull. European Assoc. Theor. Comput. Sci.}, No.\ 93, (October
  2007), 118--133.

\bibitem{Han&Wang&Wood:2006}
Y.-S. Han, Y.~Wang, and D.~Wood.
\newblock Infix-free regular expressions and languages.
\newblock {\em Internat. J. Found. Comp. Sci.} {\bf 17} (2006), 379--393.

\bibitem{Hopcroft&Ullman:1979}
J.~E. Hopcroft and J.~D. Ullman.
\newblock {\em Introduction to Automata Theory, Languages, and Computation}.
\newblock Addison-Wesley, 1979.

\bibitem{Jurgensen&Konstantinidis:1997}
H.~{J\"urgensen} and S.~Konstantinidis.
\newblock Codes.
\newblock In G.~Rozenberg and A.~Salomaa, editors, {\em Handbook of Formal
  Languages, Vol. 1}, pp.  511--607. Springer-Verlag, 1997.

\bibitem{Thierrin:1973}
G.~Thierrin.
\newblock Convex languages.
\newblock In M.~Nivat, editor, {\em Automata, Languages, and Programming}, pp.
  481--492. North-Holland, 1973.

\end{thebibliography}
\end{document}